\documentclass[draftcls,onecolumn,a4paper]{IEEEtran}

\usepackage[english]{babel}
\usepackage{amssymb,amsmath,amscd,amsfonts,amsthm,bbm,bm}
\usepackage{graphicx}
\usepackage{psfrag}
\usepackage{paralist}
\usepackage{color}
\usepackage{float}
\usepackage{relsize}
\usepackage{euscript}
\usepackage{setspace}
\usepackage{flushend}
\usepackage{tikz}
\usetikzlibrary{plotmarks}
\usepackage{pgfplots}
\usetikzlibrary{arrows,shapes,decorations,backgrounds}
\usepackage{tkz-berge}

\DeclareMathOperator{\tr}{Tr}

\DeclareMathOperator{\diag}{diag}
\newcommand{\bs}{\boldsymbol}

\newcommand{\N}{\mathbb N}
\newcommand{\R}{\mathbb R}
\newcommand{\Z}{\mathbb Z}
\newcommand{\C}{\mathbb{C}}
\newcommand{\E}{\mathbb{E}}
\newcommand{\PP}{\mathbb{P}}

\newcommand{\norme}[1]{\left\Vert #1\right\Vert}

\newtheorem{lemma}{Lemma}

\newtheorem{theorem}{Theorem}

\newtheorem{remark}{Remark}
\newtheorem{assumption}{Assumption}

\usepackage{soul}

\begin{document}

\title{Estimation of Toeplitz Covariance Matrices in \\ Large Dimensional Regime with Application\\ to Source Detection}

\author{Julia Vinogradova, Romain Couillet, and Walid Hachem
\thanks{The first and third authors are with CNRS LTCI; Telecom ParisTech, 
France (e-mail: julia.vinogradova@telecom-paristech.fr; walid.hachem@telecom-paristech.fr). The second author is with Supelec, France (e-mail: romain.couillet@supelec.fr). 
This work is supported in part by the French Ile-de-France region, DIM LSC fund, Digiteo project DESIR, and in part by the ANR-12-MONU-OOO3 DIONISOS.
}}

\date{\today} 
  
\maketitle

\def\herm{{\sf H}}
\def\asto{\overset{\rm a.s.}{\longrightarrow}}

\begin{abstract}
In this article, we derive concentration inequalities for the spectral norm of two classical sample estimators of large dimensional Toeplitz covariance matrices, demonstrating in particular their asymptotic almost sure consistence. The consistency is then extended to the case where the aggregated matrix of time samples is corrupted by a rank one (or more generally, low rank) matrix. As an application of the latter, the problem of source detection in the context of large dimensional sensor networks within a temporally correlated noise environment is studied. As opposed to standard procedures, this application is performed online, \textit{i.e.} without the need to possess a learning set of pure noise samples. 
\end{abstract}
\vspace{.3mm}

\section{Introduction}

Let $(v_{t})_{t\in\Z}$ be a complex circularly symmetric Gaussian stationary process with zero mean and covariance function $(r_k)_{k\in\Z}$ with $r_k=\E[v_{t+k}v^*_{t}]$ and $r_k\to0$ as $k\to\infty$. We observe $N$ independent copies of $(v_t)_{t\in\Z}$ over the time window $t\in\{0,\ldots,T-1\}$, and stack the observations in a matrix  $V_T=[ v_{n,t} ]_{n,t = 0}^{N-1, T-1}$. This matrix can be written as $V_T=W_TR_T^{1/2}$, where $W_T\in\C^{N\times T}$ has independent $\mathcal{CN}(0,1)$ (standard circularly symmetric complex Gaussian) entries and $R_T^{1/2}$ is any square root of the Hermitian nonnegative definite Toeplitz $T\times T$ matrix
\begin{equation*}
R_T \triangleq \left[ r_{i-j} \right]_{0\leq i,j\leq T-1} = \begin{bmatrix}
r_0 & r_{1} & \ldots & r_{T-1} \\ 
r_{-1} & \ddots & \ddots & \vdots \\ 
\vdots  & \ddots & \ddots & r_{1}\\ 
r_{1-T} & \ldots & r_{-1} & r_0
\end{bmatrix}.
\end{equation*}
A classical problem in signal processing is to estimate $R_T$ from the observation of $V_T$. 
With the growing importance of multi-antenna array processing, there has recently been a renewed interest for this estimation problem in the regime of large system dimensions, {\it i.e.} for both $N$ and $T$ large. 

At the core of the various estimation methods for $R_T$ are the biased and unbiased estimates $\hat{r}_{k,T}^b$ and $\hat{r}_{k,T}^u$ for $r_k$, respectively, defined by
\begin{align*}
	\hat{r}_{k,T}^b &= \frac{1}{NT}\sum_{n=0}^{N-1}\sum_{t=0}^{T-1} v_{n,t+k} v_{n,t}^* \mathbbm{1}_{0 \leq t+k \leq T-1} \\
	\hat{r}_{k,T}^u &= \frac{1}{N(T-|k|)} \sum_{n=0}^{N-1} \sum_{t=0}^{T-1} v_{n,t+k} v_{n,t}^* \mathbbm{1}_{0 \leq t+k \leq T-1}
\end{align*}
where $\mathbbm{1}_A$ is the indicator function on the set $A$.
Depending on the relative rate of growth of $N$ and $T$, the matrices $\widehat{R}_{T}^b = [\hat{r}_{i-j,T}^b]_{0 \leq i,j \leq T-1}$ and $\widehat{R}_{T}^u = [\hat{r}_{i-j,T}^u]_{0 \leq i,j \leq T-1}$ may not satisfy $\Vert R_T - \widehat{R}_{T}^b \Vert \asto 0$ or $\Vert R_T - \widehat{R}_{T}^u \Vert \asto 0$. An important drawback of the biased entry-wise estimate lies in its inducing a general asymptotic bias in $\widehat{R}_{T}^b$; as for the unbiased entry-wise estimate, it may induce too much inaccuracy in the top-right and bottom-left entries of $\widehat{R}_{T}^u$. The estimation paradigm followed in the recent literature generally consists instead in building banded or tapered versions of $\widehat{R}_{T}^b$ or $\widehat{R}_{T}^u$ ({\it i.e.} by weighting down or discarding a certain number of entries away from the diagonal), exploiting there the rate of decrease of $r_k$ as $k\to\infty$ \cite{WuPour'09,BickLev'08a,LamFan'09,XiaoWu'12,CaiZhangZhou'10,CaiRenZhou'13}.
Such estimates use the fact that $\Vert R_T-R_{\gamma(T),T}\Vert \to 0$ with $R_{\gamma,T} = [ [R_T]_{i,j} \mathbbm{1}_{|i-j| \leq \gamma}]$ for some well-chosen functions $\gamma(T)$ (usually satisfying $\gamma(T)\to \infty$ and $\gamma(T)/T\to 0$) and restrict the study to the consistent estimation of $R_{\gamma(T),T}$. The aforementioned articles concentrate in particular on choices of functions $\gamma(T)$ that ensure optimal rates of convergence of $\Vert R_T-\widehat{R}_{\gamma(T),T}\Vert$ for the banded or tapered estimate $\widehat{R}_{\gamma(T),T}$.
These procedures, although theoretically optimal, however suffer from several practical limitations. First, they assume the {\it a priori} knowledge of the rate of decrease of $r_k$ (and restrict these rates to specific classes). Then, even if this were indeed known in practice, being asymptotic in nature, the results do not provide explicit rules for selecting $\gamma(T)$ for practical finite values of $N$ and $T$. Finally, the operations of banding and tapering do not guarantee the positive definiteness of the resulting covariance estimate.

In the present article, we consider instead that the only constraint about $r_k$ is $\sum_{k=-\infty}^\infty |r_k|<\infty$ and estimate $R_T$ from the standard (non-banded and non-tapered) estimates $\widehat{R}_T^b$ and $\widehat{R}_T^u$. The consistence of these estimates, in general invalid, shall be enforced here by the choice $N,T\to\infty$ with $N/T\to c\in(0,\infty)$. This setting is more practical in applications as long as both the finite values $N$ and $T$ are sufficiently large and of similar order of magnitude.
Another context where a non banded Toeplitz rectification of the estimated 
covariance matrix leads to a consistent estimate in the spectral norm is 
studied in \cite{val-loub-icassp14}. 

Our specific contribution lies in the establishment of concentration inequalities for the random variables $\Vert R_T-\widehat{R}_T^b\Vert$ and $\Vert R_T-\widehat{R}_T^u\Vert$. It is shown specifically that, for all $x>0$, $-\log \mathbb{P}[\Vert R_T-\widehat{R}_T^b \Vert> x ]= O(T)$ and $-\log \mathbb{P}[\Vert R_T-\widehat{R}^u_T\Vert > x ]= O(T/ \log T)$. Aside from the consistence in norm, this implies as a corollary that, as long as $\limsup_T\Vert R_T^{-1}\Vert<\infty$, for $T$ large enough, $\widehat{R}^u_T$ is positive definite with outstanding probability ($\widehat{R}^b_T$ is nonnegative definite by construction).

For application purposes, the results are then extended to the case where $V_T$ is changed into $V_T+P_T$ for a rank-one matrix $P_T$. Under some conditions on the right-eigenspaces of $P_T$, we show that the concentration inequalities hold identically. The application is that of a single source detection (modeled through $P_T$) by an array of $N$ sensors embedded in a temporally correlated noise (modeled by $V_T$). To proceed to detection, $R_T$ is estimated from $V_T+P_T$ as $\widehat{R}_T^b$ or $\widehat{R}_T^u$, which is used as a whitening matrix, before applying a generalized likelihood ratio test (GLRT) procedure on the whitened observation. Simulations corroborate the theoretical consistence of the test. 

The remainder of the article is organized as follows. The concentration inequalities for both biased and unbiased estimates are exposed in Section~\ref{unperturbed}. The generalization to the rank-one perturbation model is presented in Section~\ref{sig+noise} and applied in the practical context of source detection in Section~\ref{detect}.

{\it Notations:} The superscript $(\cdot)^{\sf H}$ denotes Hermitian transpose, $\left\| X \right\|$ stands for the spectral norm for a matrix and Euclidean norm for a vector, and $\| \cdot \|_\infty$ is the $\sup$ norm of a function. The notations ${\cal N}(a,\sigma^2)$ and ${\cal CN}(a,\sigma^2)$ represent the real and complex circular Gaussian distributions with mean $a$ and variance $\sigma^2$. For $x \in \C^{m}$, $D_x=\diag (x)=\diag (x_0, \ldots, x_{m-1} )$ is the diagonal matrix having on its diagonal the elements of the vector $x$.
For $x=[x_{-(m-1)},\ldots,x_{m-1}]^{\sf T} \in \C^{2m+1}$, the matrix ${\cal T}(x) \in \C^{m \times m}$ is the Toeplitz matrix built from $x$ with entries $[{\cal T}(x)]_{i,j}=x_{j-i}$.
The notations $\Re(\cdot)$ and $\Im(\cdot)$ stand for the real and the imaginary parts respectively.

%%%%%%%%%%%%%%%%%%%%%%%%%%%%%%%%%%%%%%%
\section{Performance of the covariance matrix estimators} 
\label{unperturbed} 

\subsection{Model, assumptions, and results}
\label{subsec-model} 

%We start by precisely stating our model. 
Let $(r_k)_{k\in\Z}$ be a doubly 
infinite sequence of covariance coefficients. For any $T \in \N$, let 
$R_T = \mathcal T( r_{-(T-1)},\ldots, r_{T-1})$, a Hermitian nonnegative definite matrix. 
Given $N = N(T) > 0$, consider the 
matrix model 
\begin{equation}
V_T = [ v_{n,t} ]_{n,t = 0}^{N-1, T-1}  = W_T R_T^{1/2}
\label{model1}
\end{equation}
where $W_T = [ w_{n,t} ]_{n,t = 0}^{N-1, T-1}$ has independent ${\cal CN}(0,1)$ entries. 
It is clear that $r_k = \E [v_{n,t+k} v_{n,t}^*]$  for any $t$, $k$, and $n \in \{0,\ldots, N-1\}$.

In the following, we shall make the two assumptions below.
\begin{assumption} 
\label{ass-rk} 
The covariance coefficients $r_k$ are absolutely summable and $r_0 \neq 0$. 
\end{assumption} 
With this assumption, the covariance function 
\[
{\bs\Upsilon}(\lambda) \triangleq 
\sum_{k=-\infty}^\infty r_k e^{-\imath k\lambda}, \quad 
\lambda \in [0, 2\pi) 
\] 
is continuous on the interval $[0, 2\pi]$. Since $\| R_T \| \leq \| \bs\Upsilon \|_\infty$ (see \emph{e.g.} \cite[Lemma 4.1]{Gray'06}), Assumption \ref{ass-rk} 
implies that $\sup_T \| R_T \| < \infty$.  

We assume the following asymptotic regime which will be simply denoted as ``$T\rightarrow\infty$'':
\begin{assumption} 
\label{ass-regime} 
$T \rightarrow \infty$ and $N/T \rightarrow c > 0$.
\end{assumption} 

%The subscripts and/or superscripts $T$ will be often omitted for notational 
%simplicity. 

Our objective is to study the performance of two estimators of the 
covariance function frequently considered in the literature. These estimators are defined as 
\begin{align}
\label{est-b} 
\hat{r}_{k,T}^b&= 
\frac{1}{NT} \sum_{n=0}^{N-1}\sum_{t=0}^{T-1}
v_{n,t+k} v_{n,t}^* \mathbbm{1}_{0 \leq t+k \leq T-1} \\
\label{est-u} 
\hat{r}_{k,T}^u&=\frac{1}{N(T-|k|)}
\sum_{n=0}^{N-1}\sum_{t=0}^{T-1} v_{n,t+k} v_{n,t}^* 
\mathbbm{1}_{0 \leq t+k \leq T-1}.
\end{align}
Since  $\E \hat{r}_{k,T}^b = ( 1 - |k|/T ) r_k$ and $\E \hat{r}_{k,T}^u = r_k$,
the estimate $\hat{r}_{k,T}^b$ is biased while $\hat{r}_{k,T}^u$ is unbiased. 
Let also 
\begin{align}
	\label{est-Rb}
\widehat R^b_T &\triangleq {\cal T}\left( \hat{r}_{-(T-1),T}^b, \ldots, 
\hat{r}_{(T-1),T}^b \right) \\
	\label{est-Ru}
\widehat R^u_T &\triangleq {\cal T} \left( \hat{r}_{-(T-1),T}^u, \ldots, 
\hat{r}_{(T-1),T}^u \right).
\end{align}
A well known advantage of $\widehat R^b_T$ over $\widehat R^u_T$ as an estimate of $R_T$ is its structural nonnegative definiteness. % (see \emph{e.g.} \cite[Section 5.3.3]{Priestley'81}). 
In this section, results on the spectral behavior of these matrices are provided under the form of concentration inequalities on $\| \widehat R^b_T - R_T \|$ and $\| \widehat R^u_T - R_T \|$: 
\begin{theorem}
\label{th-biased} 
Let Assumptions~\ref{ass-rk} and \ref{ass-regime} hold true and let $\widehat R^b_T$ be defined as in \eqref{est-Rb}.
%Consider the biased estimates $\hat{r}_{k,T}^b$ and let $\widehat R_{T}^b = {\cal T}( \hat{r}_{-(T-1),T}^b, \ldots, \hat{r}_{(T-1),T}^b)$. 
Then, for any $x>0$,
\begin{equation*}
\mathbb{P} \left[\norme{\widehat R_T^b - R_T} > x \right] \leq
\exp \left( -cT \left( \frac{x}{\| \bs\Upsilon \|_\infty} - 
 \log \left( 1 + \frac{x}{\| \bs\Upsilon \|_\infty} \right) + o(1) \right)
\right)
\end{equation*}
where $o(1)$ is with respect to $T$ and depends on $x$.
\end{theorem}
\begin{theorem}
\label{th-unbiased} 
Let Assumptions~\ref{ass-rk} and \ref{ass-regime} hold true and let $\widehat R^u_T$ be defined as in \eqref{est-Ru}.
Then, for any $x>0$,
\begin{equation*}
\mathbb{P} \left[\norme{\widehat R_T^u - R_T} > x \right] 
\leq 
\exp \left(- \frac{cTx^2}{4\norme{\bs\Upsilon}_{\infty}^2\log T} (1 + o(1)) \right) 
\end{equation*}
where $o(1)$ is with respect to $T$ and depends on $x$.
\end{theorem}

A consequence of these theorems, obtained by the Borel-Cantelli lemma, is that $\| \widehat R^b_T - R_T \| \to 0$ and 
$\| \widehat R^u_T - R_T \| \to 0$ almost surely as $T \to \infty$. 

The slower rate of decrease of $T/\log(T)$ in the unbiased estimator exponent may be interpreted by the increased inaccuracy in the estimates of $r_k$ for values of $k$ close to $T-1$. 

We now turn to the proofs of Theorems \ref{th-biased} and \ref{th-unbiased},
starting with some basic mathematical results that will be needed throughout 
the proofs. 

\subsection{Some basic mathematical facts} 

\begin{lemma}
\label{lm-fq} 
For $x,y \in \C^{m}$ and $A \in \C^{m\times m}$,
\[
\left| x^{\sf H} A x - y^{\sf H} A y  \right| \leq 
\norme{A} (\norme{x}+\norme{y})\norme{x-y}.
\]
\end{lemma}
\begin{proof}
\begin{align*} 
\left| x^{\sf H} A x - y^{\sf H} A y  \right| &= 
\left| x^{\sf H} A x - y^{\sf H} A x + y^{\sf H} A x - y^{\sf H} A y  \right| \\
&\leq \left| (x - y)^{\sf H} A x \right| + \left| y^{\sf H} A (x - y) \right| \\
&\leq \norme{A}(\norme{x} + \norme{y}) \norme{x-y}.
\end{align*}
\end{proof}

\begin{lemma} 
\label{chernoff} 
Let $X_0, \ldots, X_{M-1}$ be independent $\mathcal{CN}(0,1)$ random 
variables. Then, for any $x > 0$, 
\[
\mathbb{P} \left[ \frac1M \sum_{m=0}^{M-1} (|X_m|^2 - 1) > x \right] 
\leq \exp \left( -M ( x - \log(1+x) ) \right) . 
\] 
\end{lemma}
 
\begin{proof} 
This is a classical Chernoff bound. Indeed, given $\xi \in (0,1)$, we have
by the Markov inequality 
\begin{align*}
\mathbb{P} \Bigl[ M^{-1} \sum_{m=0}^{M-1} (|X_m|^2 - 1) > x \Bigr] &= \mathbb{P}\left[ \exp \left( \xi \sum_{m=0}^{M-1} |X_m|^2 \right) 
> \exp \xi M (x+1) \right] \\
&\leq \exp(-\xi M (x+1)) 
\E \left[ \exp \left( \xi \sum_{m=0}^{M-1} |X_m|^2 \right) \right] \\
&= \exp \left(- M \left( \xi(x+1) + \log(1-\xi) \right) \right) 
\end{align*} 
since $\E \left[\exp (\xi |X_m|^2) \right] = 1/(1-\xi)$. The result follows upon 
minimizing this expression with respect to $\xi$. 
\end{proof}

\subsection{Biased estimator: proof of Theorem \ref{th-biased}} \label{biased}
Define
\begin{align*}
\widehat \Upsilon^b_T(\lambda) &\triangleq  
\sum_{k=-(T-1)}^{T-1} \hat{r}_{k,T}^b e^{\imath k \lambda} 
\\
%\quad \text{and} \quad  
\Upsilon_T(\lambda) &\triangleq  
\sum_{k=-(T-1)}^{T-1} r_k e^{\imath k \lambda}.
\end{align*} 
Since $\widehat{R}_T^b-R_T$ is a Toeplitz matrix, from \cite[Lemma 4.1]{Gray'06}, 
\begin{equation*}
\norme{\widehat{R}_T^b-R_T} \leq 
\underset{\lambda\in[0,2\pi)}{\text{sup}} 
\left| \widehat \Upsilon_T^b(\lambda) - \Upsilon_T(\lambda) \right| 
\leq 
\underset{\lambda\in[0,2\pi)}{\text{sup}} 
\left| \widehat \Upsilon_T^b(\lambda) - \E \widehat \Upsilon_T^b(\lambda) \right| +
\underset{\lambda\in[0,2\pi)}{\text{sup}} 
\left| \E \widehat \Upsilon_T^b(\lambda) - \Upsilon_T(\lambda) \right|. 
\end{equation*}
By Kronecker's lemma (\cite[Lemma 3.21]{Kallenberg'97}), the rightmost term at the right-hand side satisfies 
\begin{equation}
\label{determ-biased} 
\left| \E \widehat \Upsilon_T^b(\lambda) - \Upsilon_T(\lambda)\right| 
\leq  \sum_{k=-(T-1)}^{T-1} \frac{|k r_k|}{T}  
\xrightarrow[T\to\infty]{} 0. 
\end{equation} 
In order to deal with the term 
$\sup_{\lambda \in [0,2\pi)} | \widehat \Upsilon_T^b(\lambda) - \E \widehat \Upsilon_T^b(\lambda) |$, 
two ingredients will be used. The first one is the 
following lemma (proven in Appendix \ref{anx-lm-qf}):

\begin{lemma}
\label{lemma_d_quad}
The following facts hold:
\begin{align*} 
\widehat \Upsilon_T^b(\lambda) &= d_T(\lambda)^{\sf H}\frac{V_T^{\sf H}V_T}{N}d_T(\lambda) \\
\E \widehat \Upsilon_T^b(\lambda) &= d_T(\lambda)^{\sf H} R_T d_T(\lambda)
\end{align*} 
where $d_T(\lambda)=1/\sqrt{T}\left[1, e^{- \imath\lambda}, \ldots, 
e^{-\imath(T-1)\lambda} \right]^{\sf T}$.
\end{lemma}

The second ingredient is a Lipschitz property of the function 
$\| d_T(\lambda) - d_T(\lambda') \|$ seen as a function of $\lambda$. 
From the inequality $|e^{-\imath t\lambda}-e^{-\imath t\lambda'}| 
\leq t|\lambda-\lambda'|$, we indeed have
\begin{equation} 
\label{lipschitz} 
\| d_T(\lambda) - d_T(\lambda') \| = 
\sqrt{\frac{1}{T} \sum_{t=0}^{T-1} |e^{-\imath t\lambda}-
e^{-\imath t\lambda'}|^2} \leq \frac{T|\lambda-\lambda'|}{\sqrt{3}} . 
\end{equation}

Now, denoting by $\lfloor \cdot \rfloor$ the floor function and choosing 
$\beta > 2$, define ${\cal I}=\left\{0, \ldots, \lfloor T^{\beta} \rfloor - 1 \right\}$.
Let $\lambda_i=2 \pi \frac{i}{\lfloor T^{\beta} \rfloor }$, $i \in {\cal I}$, be a regular discretization of the 
interval $[0, 2\pi]$. 
We write 
\begin{align*}
&\underset{\lambda \in [0, 2\pi)}{\text{sup}} \left| \widehat \Upsilon_T^b(\lambda) - \E \widehat \Upsilon_T^b(\lambda) \right| 
 \\&\leq 
\underset{i \in {\cal I}}{\text{max}} 
\underset{\lambda \in [\lambda_i,\lambda_{i+1}]}{\text{sup}} \Bigl(\left| \widehat \Upsilon_T^b(\lambda) 
- \widehat \Upsilon_T^b(\lambda_i) \right| + \left| \widehat \Upsilon_T^b(\lambda_i) 
- \E \widehat \Upsilon_T^b(\lambda_i) \right| + \left| \E \widehat \Upsilon_T^b(\lambda_i) 
- \E \widehat \Upsilon_T^b(\lambda) \right|\Bigr)\\ &\leq 
\underset{i \in {\cal I}}{\text{max}} 
\underset{\lambda \in [\lambda_i,\lambda_{i+1}]}{\text{sup}} 
\left| \widehat \Upsilon_T^b(\lambda) - \widehat \Upsilon_T^b(\lambda_i) \right| + \underset{i \in {\cal I}}{\text{max}} \left| \widehat \Upsilon_T^b(\lambda_i) 
- \E \widehat \Upsilon_T^b(\lambda_i) \right| + \underset{i \in {\cal I}}{\text{max}} 
\underset{\lambda \in [\lambda_i,\lambda_{i+1}]}{\text{sup}} 
\left| \E \widehat \Upsilon_T^b(\lambda_i) - \E \widehat \Upsilon_T^b(\lambda) \right| \\ &\triangleq \chi_1 + \chi_2 + \chi_3 . 
\end{align*} 
With the help of Lemma \ref{lemma_d_quad} and \eqref{lipschitz}, 
we shall provide concentration inequalities on the random terms $\chi_1$ 
and $\chi_2$ and a bound on the deterministic term $\chi_3$. 
This is the purpose of the three following lemmas. Herein and in the 
remainder, $C$ denotes a positive constant independent of $T$. This constant 
can change from an expression to another. 

\begin{lemma}
\label{chi1} 
There exists a constant $C > 0$ such that for any $x > 0$ and any $T$ large
enough, 
\begin{equation*}
\mathbb{P} \left[ \chi_1 > x \right]
\leq \exp \Biggl( - cT^2\Biggl( \frac{x T^{\beta-2}}{C \| \bs\Upsilon\|_\infty} 
- \log \frac{x T^{\beta-2}}{C \| \bs\Upsilon\|_\infty} - 1 \Biggr) \Biggr) . 
\end{equation*}
\end{lemma} 

\begin{proof}
Using Lemmas \ref{lemma_d_quad} and \ref{lm-fq} along with 
\eqref{lipschitz}, we have 
\begin{align*} 
\left| \widehat \Upsilon_T^b(\lambda) - \widehat \Upsilon_T^b(\lambda_i) \right|
&= \left| d_T(\lambda)^{\sf H}\frac{V_T^{\sf H} V_T}{N} d_T(\lambda) 
        - d_T(\lambda_i)^{\sf H} \frac{V_T^{\sf H} V_T}{N}d_T(\lambda_i) \right| \\
&\leq 2 N^{-1} \norme{d_T(\lambda) - d_T(\lambda_i)} \| R_T\| \norme{W_T^{\sf H} W_T} \\
&\leq C | \lambda - \lambda_i | \| \bs\Upsilon\|_\infty \norme{W_T^{\sf H} W_T} . 
\end{align*} 
From $\| W_T^{\sf H} W_T \| \leq \tr(W_T^{\sf H} W_T)$ and 
Lemma~\ref{chernoff}, assuming $T$ large enough so that $f(x,T) \triangleq x T^{\beta-1} / (C N \| \bs\Upsilon\|_\infty)$ satisfies 
$f(x,T) \geq 1$, we then obtain  
\begin{align*}
\mathbb{P} \left[ \chi_1 > x \right] &\leq 
\mathbb{P} \left[ 
C \| \bs\Upsilon\|_\infty T^{-\beta} 
\sum_{t=0}^{T-1}\sum_{n=0}^{N-1} |w_{n,t}|^2 > x \right] \\
&= \mathbb{P} \left[ 
\frac{1}{NT} \sum_{n,t} (| w_{n,t} |^2 -1 ) > f(x,T) - 1 \right] \\
&\leq \exp( - NT( f(x,T) - \log f(x,T) - 1 ) ) . 
\end{align*} 
\end{proof}

\begin{lemma} 
\label{chi2} 
The following inequality holds
\begin{equation*}
\mathbb{P} \left[ \chi_2 > x \right] \leq
2T^{\beta} \exp \Biggl( - c T \Biggl( \frac{x}{\|\bs\Upsilon\|_\infty} - 
\log \Bigl( 1 + \frac{x}{\|\bs\Upsilon\|_\infty} \Bigr) \Biggr) \Biggr). 
\end{equation*}
\end{lemma} 

\begin{proof}
From the union bound we obtain:
\begin{align*} 
\mathbb{P} \left[ \chi_2 > x \right]
&\leq \sum_{i=0}^{\lfloor T^{\beta} \rfloor - 1} 
\mathbb{P} \left[ \left| \widehat \Upsilon_T^b(\lambda_i) 
 - \E \widehat \Upsilon_T^b(\lambda_i) \right| > x \right].
\end{align*} 
We shall bound each term of the sum separately. Since 
\begin{equation*}  
\mathbb{P} \left[ \left| \widehat \Upsilon_T^b(\lambda_i) 
              - \E \widehat \Upsilon_T^b(\lambda_i) \right| > x \right] 
= \mathbb{P} \left[ \widehat \Upsilon_T^b(\lambda_i) 
                 - \E \widehat \Upsilon_T^b(\lambda_i) > x \right] +
\mathbb{P} \left[ - \left( \widehat \Upsilon_T^b(\lambda_i) - \E \widehat \Upsilon_T^b(\lambda_i) \right) > x \right] 
\end{equation*}
it will be enough to deal with the first right-hand side term as the second 
one is treated similarly.
Let $\eta_T(\lambda_i) \triangleq W_T q_T(\lambda_i) = 
\left[ \eta_{0,T}(\lambda_i), \ldots, \eta_{N-1,T}(\lambda_i) \right]^{\sf T}$ 
where $q_T(\lambda_i) \triangleq R_T^{1/2} d_T(\lambda_i)$. Observe that 
$\eta_{k,T}(\lambda_i) \sim \mathcal{CN}(0,\| q_T(\lambda_i) \|^2 I_N)$. We know from Lemma~\ref{lemma_d_quad} that 
\begin{equation}
\widehat \Upsilon_T^b(\lambda_i) - \E \widehat \Upsilon_T^b(\lambda_i) = 
\frac 1N \left( \| \eta_T(\lambda_i) \|^2 - \E \| \eta_T(\lambda_i) \|^2 \right).
\label{Epsilon_eta}
\end{equation}
From (\ref{Epsilon_eta}) and Lemma~\ref{chernoff}, we therefore get  
\begin{equation*} 
\mathbb{P}\left[ \widehat \Upsilon_T^b(\lambda_i) - \E \widehat \Upsilon_T^b(\lambda_i) > x \right]
\leq \exp \Biggl( -N \Biggl( \frac{x}{\| q_T(\lambda_i) \|^2} - 
 \log\Bigl( 1 + \frac{x}{\| q_T(\lambda_i) \|^2} \Bigr) \Biggr) \Biggr).  
\end{equation*} 
Noticing that $\| q_T(\lambda_i) \|^2 \leq \|\bs\Upsilon\|_\infty$ and that the function $f(x) = x - \log \Bigl( 1 + x \Bigr)$ is increasing for $x>0$, we get the result.
\end{proof} 

Finally, the bound for the deterministic term $\chi_3$ is provided by the following lemma:
\begin{lemma}
\label{chi3} 
$\displaystyle{ 
\chi_3 \leq C \| \bs\Upsilon\|_\infty T^{-\beta + 1}
}$. 
\end{lemma} 
\begin{proof} 
From Lemmas \ref{lemma_d_quad} and \ref{lm-fq} along with 
\eqref{lipschitz}, we obtain
\begin{align*} 
\left| \E \widehat \Upsilon_T^b(\lambda) - \E \widehat \Upsilon_T^b(\lambda_i) \right|
&= \left| d_T(\lambda)^{\sf H} R_T d_T(\lambda) - d_T(\lambda_i)^{\sf H} R_T d_T(\lambda_i) \right| \\
&\leq 2 \norme{R_T} \norme{d_T(\lambda) - d_T(\lambda_i)} \\
&\leq C \| \bs\Upsilon\|_\infty | \lambda - \lambda_i | T. 
\end{align*} 
From $\underset{i \in {\cal I}}{\text{max}} \underset{\lambda \in [\lambda_i,\lambda_{i+1}]}{\text{sup}} 
| \lambda - \lambda_i | = \lambda_{i+1} - \lambda_i = T^{-\beta}$ we get the result.
\end{proof}

We now complete the proof of Theorem \ref{th-biased}. From 
\eqref{determ-biased} and Lemma~\ref{chi3}, we get
\begin{equation*}
\mathbb{P} \left[\norme{\widehat R_T^b - R_T} > x \right] = 
\mathbb{P}\left[ \chi_1 + \chi_2 > x + o(1) \right]. 
\end{equation*} 
Given a parameter $\epsilon_T \in [0,1]$, we can write (with some slight notation abuse)
\begin{equation*}
\mathbb{P}\left[ \chi_1 + \chi_2 > x + o(1) \right] \leq
\mathbb{P}\left[ \chi_1 > x \epsilon_T \right] + \mathbb{P}\left[\chi_2 > x (1 - \epsilon_T) + o(1) \right].
\end{equation*}
With the results of Lemmas \ref{chi1} and \ref{chi2}, setting $\epsilon_T=1/T$, we get
\begin{align*}
\mathbb{P}\left[ \chi_1 + \chi_2 > x + o(1) \right] &\leq
\mathbb{P}\left[ \chi_1 > \frac{x}{T} \right] + \mathbb{P}\left[ \chi_2 > x (1 - \frac{x}{T}) + o(1) \right] \\
&\leq
\exp \Bigl( - cT^2 \Bigl( \frac{x T^{\beta-3}}{C \| \bs\Upsilon\|_\infty} 
- \log \frac{x T^{\beta-3}}{C \| \bs\Upsilon\|_\infty} - 1 \Bigr) \Bigr) \\
&+ \exp \Bigl( - cT \Bigl( \frac{x\left( 1 - \frac{1}{T} \right)}{\|\bs\Upsilon\|_\infty} - 
\log \Bigl( 1 + \frac{x\left( 1 - \frac{1}{T} \right)}{\|\bs\Upsilon\|_\infty} \Bigr) + o(1) \Bigr) \Bigr) \\
&=
\exp \Bigl( - cT \Bigl( \frac{x}{\|\bs\Upsilon\|_\infty} - 
\log \Bigl( 1 + \frac{x}{\|\bs\Upsilon\|_\infty} \Bigr) + o(1) \Bigr) \Bigr) 
\end{align*}
since $\beta>2$.

\subsection{Unbiased estimator: proof of Theorem \ref{th-unbiased}}\label{unbiased}
The proof follows basically the same main steps as for Theorem~\ref{th-biased} 
with an additional difficulty due to the scaling terms $1/(T-|k|)$.

Defining the function 
\[ 
\widehat \Upsilon^u_T(\lambda) \triangleq  
\sum_{k=-(T-1)}^{T-1} \hat{r}_{k,T}^u e^{ik \lambda}
\]
we have
\begin{equation*} 
\norme{\widehat{R}_T^u-R_T} \leq 
\underset{\lambda\in[0,2\pi)}{\text{sup}} 
\left| \widehat \Upsilon_T^u(\lambda) - \Upsilon_T(\lambda) \right| = 
\underset{\lambda\in[0,2\pi)}{\text{sup}} 
\left| \widehat \Upsilon_T^u(\lambda) - 
    \E \widehat \Upsilon_T^u(\lambda) \right|
\end{equation*} 
since $\Upsilon_T(\lambda)=\E \widehat \Upsilon_T^u(\lambda)$, the estimates
$\hat{r}_{k,T}^u$ being unbiased. 

In order to deal with the right-hand side of this expression, we need the 
following analogue of Lemma~\ref{lemma_d_quad}, borrowed from 
\cite{val-loub-icassp14} and proven here in Appendix~\ref{anx-lm-qf2}.
\begin{lemma}
\label{lemma_d_quad2}
The following fact holds:
\begin{align*}
\widehat \Upsilon_T^u(\lambda) &= d_T(\lambda)^{\sf H} 
\left( \frac{V_T^{\sf H}V_T}{N} \odot B_T \right) d_T(\lambda)
% \E \widehat \Upsilon_T^u(\lambda) &= 
% d_T(\lambda)^{\sf H} \left( R_T \odot B_T \right) d_T(\lambda)
\end{align*} 
where $\odot$ is the Hadamard product of matrices and where 
\[
	B_T \triangleq \left[ \frac{T}{T-|i-j|} \right]_{0\leq i,j\leq T-1}.
%	= \begin{bmatrix}
%1 & \frac{T}{T-1} & \ldots & T \\ 
%\frac{T}{T-1} & \ddots & \ddots & \vdots \\ 
%\vdots  & \ddots & \ddots & \frac{T}{T-1}\\ 
%T & \ldots & \frac{T}{T-1} & 1
%\end{bmatrix}.
\] 
\end{lemma}

In order to make $\widehat \Upsilon_T^u(\lambda)$ more tractable, we rely on the 
following lemma which can be proven by direct calculation.
\begin{lemma}
\label{lemma_hadamard}
Let $x$, $y \in \C^{m}$ and $A, B \in C^{m \times m}$. Then 
\begin{equation*}
x^{\sf H}( A \odot B ) y = \tr (D_x^{\sf H} A D_y B^{\sf T}) 
\end{equation*}
where we recall $D_x = \diag(x)$ and $D_y = \diag(y)$.
\end{lemma}

Denoting
\begin{align*}
D_T(\lambda) &\triangleq \diag (d_T(\lambda)) = \frac{1}{\sqrt{T}} \diag(1, e^{i\lambda}, \ldots, e^{i(T-1)\lambda} ) \\
Q_T(\lambda) &\triangleq R_T^{1/2} D_T(\lambda) B_T D_T(\lambda)^{\sf H} (R_T^{1/2})^{\sf H}
\end{align*}
we get from Lemmas~\ref{lemma_d_quad2} and \ref{lemma_hadamard} 
\begin{align}
\label{Upsilon_sum} 
\widehat \Upsilon_T^u(\lambda) &= \frac1N
\tr(D_T(\lambda)^{\sf H} (R_T^{1/2})^{\sf H} W_T^{\sf H} W_T R_T^{1/2} D_T(\lambda) B_T)
\nonumber\\ &= \frac1N \tr (W_T Q_T(\lambda) W_T^{\sf H}) \nonumber \\
&= \frac1N
\sum_{n=0}^{N-1} w_n^{\sf H} Q_T(\lambda) w_n 
\end{align} 
where $w_i^{\sf H}$ is such that $W_T=[w_0^{\sf H},\ldots,w_{N-1}^{\sf H}]$. 

Compared to the biased case, the main difficulty lies here in the fact 
that the matrices $B_T/T$ and $Q_T(\lambda)$ have unbounded spectral norm as $T\to\infty$. 
The following lemma, proven in Appendix~\ref{prf-lm-B-Q}, provides some information on the spectral behavior of these
matrices that will be used subsequently.  
\begin{lemma}
\label{lm-B-Q} 
The matrix $B_T$ satisfies 
\begin{equation}
\label{norme_B} 
\norme{B_T} \leq \sqrt{2} T( \sqrt{\log T} + C). 
\end{equation} 
For any $\lambda \in[0, 2\pi)$, the eigenvalues 
$\sigma_0, \ldots, \sigma_{T-1}$ of the matrix $Q(\lambda)$ satisfy the 
following inequalities: 
\begin{eqnarray}
\sum_{t=0}^{T-1} \sigma_t^2 &\leq& 2 \norme{\bs\Upsilon}_{\infty}^2 \log T + C \label{sum_sigma2}\\ 
\underset{t}{\max} |\sigma_t| &\leq& \sqrt{2} \| \bs \Upsilon \|_\infty 
( \log T )^{1/2} + C \label{sig_max} \\
\sum_{t=0}^{T-1} |\sigma_t|^3 &\leq& C ((\log T)^{3/2} +1)\label{sum_sigma3}
\end{eqnarray}
where the constant $C$ is independent of $\lambda$.
\end{lemma} 

We shall also need the following easily shown Lipschitz property of the 
function $\norme{D_T(\lambda) - D_T(\lambda')}$:  
\begin{equation} 
\label{lipschitz_D} 
\| D_T(\lambda) - D_T(\lambda') \| \leq \sqrt{T}|\lambda - \lambda'|. 
\end{equation}

We now enter the core of the proof of Theorem~\ref{th-unbiased}.  
Choosing $\beta>2$, let $\lambda_i=2 \pi \frac{i}{\lfloor T^{\beta}\rfloor}$, 
$i \in {\cal I}$, be a regular discretization of the interval $[0, 2\pi]$ with 
${\cal I}=\left\{0, \ldots, \lfloor T^{\beta} \rfloor - 1 \right\}$. We write 
\begin{align*}
\underset{\lambda \in [0, 2\pi)}{\text{sup}} 
\left| \widehat \Upsilon_T^u(\lambda) - \E \widehat \Upsilon_T^u(\lambda) \right|  
&\leq \underset{i \in {\cal I}}{\text{max}} 
\underset{\lambda \in [\lambda_i,\lambda_{i+1}]}{\text{sup}} 
\left| \widehat \Upsilon_T^u(\lambda)-\widehat \Upsilon_T^u(\lambda_i) \right| + \underset{i \in {\cal I}}{\text{max}} \left| \widehat \Upsilon_T^u(\lambda_i) - \E \widehat \Upsilon_T^u(\lambda_i) \right|\\& + \underset{i \in {\cal I}}{\text{max}} \underset{\lambda \in [\lambda_i,\lambda_{i+1}]}{\text{sup}} 
\left| \E \widehat \Upsilon_T^u(\lambda_i) - \E \widehat \Upsilon_T^u(\lambda) \right| \\ 
&\triangleq \chi_1 + \chi_2 + \chi_3 . 
\end{align*} 

Our task is now to provide concentration inequalities on the random terms 
$\chi_1$ and $\chi_2$ and a bound on the deterministic term $\chi_3$.

\begin{lemma}
\label{chi1-u} 
There exists a constant $C > 0$ such that, if $T$ is large enough, the following 
inequality holds:
\begin{equation*}
\mathbb{P} \left[ \chi_1 > x \right]
\leq \exp \left( - cT^2 \left( \frac{x T^{\beta-2}}{C \sqrt{\log T}} 
- \log \frac{x T^{\beta-2}}{C \sqrt{\log T}} - 1 \right) \right). 
\end{equation*} 
\end{lemma} 

\begin{proof}
From Equation~\eqref{Upsilon_sum}, we have
\begin{align} 
\left| \widehat \Upsilon_T^u(\lambda) - \widehat \Upsilon_T^u(\lambda_i) \right|
& = \frac1N \left| \sum_{n=0}^{N-1} w_n^{\sf H}\left( Q_T(\lambda) - Q_T(\lambda_i) \right) w_n \right| \nonumber \\
& \leq \frac1N \sum_{n=0}^{N-1} \left| w_n^{\sf H} \left( Q_T(\lambda) - Q_T(\lambda_i) \right) w_n \right| \nonumber\\
& \leq \frac1N \norme{Q_T(\lambda) - Q_T(\lambda_i)} \sum_{n=0}^{N-1} \norme{w_n}^2 \nonumber.
\end{align}
%where the equality is obtained by using \eqref{Upsilon_sum}, the inequality (a) is get from the triangle inequality, and (b) is
The norm above further develops as
\begin{align*} 
&\norme{Q_T(\lambda) - Q_T(\lambda_i)}\\
& \leq \norme{R_T} \Vert D_T(\lambda)B_TD_T(\lambda)^{\sf H}- 
D_T(\lambda_i)B_TD_T(\lambda)^{\sf H} + D_T(\lambda_i)B_TD_T(\lambda)^{\sf H} - 
D_T(\lambda_i)B_TD_T(\lambda_i)^{\sf H} \Vert \\
&\leq 2 \norme{D_T(\lambda)} \norme{R_T} \norme{B_T} 
\norme{D_T(\lambda) - D_T(\lambda_i)}  \leq C T ( \sqrt{\log T}  + 1)  \left| \lambda - \lambda_i \right| 
\end{align*}
where we used \eqref{norme_B}, \eqref{lipschitz_D}, and $\norme{D_T(\lambda)}=1/\sqrt{T}$.
Up to a change in $C$, we can finally write $\norme{Q_T(\lambda) - Q_T(\lambda_i)} \leq 
C T^{1-\beta} \sqrt{\log T}$. Assume that $f(x,T) \triangleq xT^{\beta-2}/
\left( C \sqrt{\log T} \right)$ 
satisfies $f(x,T) > 1$ (always possible for every fixed $x$ by taking $T$ large). Then we get by Lemma~\ref{chernoff} 
\begin{align*} 
\mathbb{P} \left[ \chi_1 > x \right] & \leq \mathbb{P} \left( C N^{-1} T^{1-\beta} 
\sqrt{\log T} \sum_{n,t} \left|w_{n,t}\right|^2 > x \right) \\ 
&= \PP \left( \frac{1}{NT} \sum_{n,t} (\left|w_{n,t}\right|^2 - 1) 
> f(x,T) - 1 \right) \\
&\leq
\exp \left( - NT \left( f(x,T) - \log \left(f(x,T)\right) - 1 \right) \right).
\end{align*}
\end{proof}

The most technical part of the proof is to control the term $\chi_2$, which we handle hereafter.

\begin{lemma}
\label{chi2-u} 
The following inequality holds:
\begin{equation*}
\mathbb{P} \left[ \chi_2 > x \right]
\leq \exp \left(- \frac{cx^2T}{4\norme{\bs\Upsilon}_{\infty}^2\log T} (1 + o(1)) \right).
\end{equation*}
\end{lemma}

\begin{proof}
From the union bound we obtain:
\begin{equation} \label{union_bound}
\mathbb{P} \left[ \chi_2 > x \right]
\leq \sum_{i=0}^{\lfloor T^{\beta} \rfloor - 1} 
\mathbb{P} \left[ \left| \widehat \Upsilon_T^u(\lambda_i) 
 - \E \widehat \Upsilon_T^u(\lambda_i) \right| > x \right].
\end{equation} 
Each term of the sum can be written
\begin{equation*}
\mathbb{P} \left[ \left| \widehat \Upsilon_T^u(\lambda_i) - \E \widehat \Upsilon_T^u(\lambda_i) \right| > x \right] = \mathbb{P} \left[\widehat \Upsilon_T^u(\lambda_i) - \E \widehat \Upsilon_T^u(\lambda_i) > x \right] + \mathbb{P} \left[- \left(\widehat \Upsilon_T^u(\lambda_i) - \E \widehat \Upsilon_T^u(\lambda_i) \right) > x \right].
\end{equation*}
We will deal with the term $\psi_i = \mathbb{P} \left[\widehat \Upsilon_T^u(\lambda_i) - \E \widehat \Upsilon_T^u(\lambda_i) > x \right]$, the term $\mathbb{P} \left[- \left(\widehat \Upsilon_T^u(\lambda_i) - \E \widehat \Upsilon_T^u(\lambda_i) \right) > x \right]$ being treated similarly.
Let $Q_T(\lambda_i) = U_T \Sigma_T U_T^{\sf H}$ be a spectral 
factorization of the Hermitian matrix $Q_T(\lambda_i)$ with 
$\Sigma_T =\diag (\sigma_{0},\ldots,\sigma_{T-1})$. 
Since $U_T$ is unitary and $W_T$ has independent ${\cal CN}(0,1)$ elements,
we get from Equation \eqref{Upsilon_sum} 
\begin{equation}
\label{Ups_u}
\widehat \Upsilon_T^u(\lambda_i) \stackrel{\cal L}{=}  
\frac{1}{N}\sum_{n=0}^{N-1} w_n^{\sf H} \Sigma_T(\lambda_i) w_n 
= \frac{1}{N}\sum_{n=0}^{N-1} \sum_{t=0}^{T-1} |w_{n,t}|^2 \sigma_t
\end{equation}
where $\stackrel{\cal L}{=}$ denotes equality in law. 
Since $\E \left[ e^{a|X|^2} \right]= 1/(1-a)$ when $X \sim {\cal CN}(0, 1)$ 
and $0<a<1$, we have by Markov's inequality and from the independence of the
variables $|w_{n,t}|^2$
\begin{align}
\psi_i &= \PP \left( 
\frac{1}{N}\sum_{n=0}^{N-1} \sum_{t=0}^{T-1} |w_{n,t}|^2 \sigma_t 
        - \tr Q_T(\lambda_i) > x \right) \nonumber \\ 
&\leq \E\left[\text{exp}\Bigl( 
\frac{\tau}{N}\sum_{n,t} |w_{n,t}|^2\sigma_t \Bigr) \right]
\exp \Bigl(-\tau\Bigl(x + \sum_{t=0}^{T-1}\sigma_t \Bigr)\Bigr) \nonumber \\ 
&= 
\exp \Bigl( -\tau \Bigl(x+ \sum_{t=0}^{T-1} \sigma_t \Bigr) \Bigr) 
\prod_{t=0}^{T-1} \Bigl( 1- \frac{\sigma_t\tau}{N} \Bigr)^{-N} 
\label{pro} \\
&= \exp\Bigl(-\tau \Bigl(x+ \sum_{t=0}^{T-1} \sigma_t \Bigr) 
  - N \sum_{t=0}^{T-1}  \log \Bigl(1-\frac{\sigma_t\tau}{N} \Bigr) \Bigr) 
\nonumber
\end{align}
for any $\tau$ such that $0 \leq \tau < \underset{0\leq t\leq T-1}{\text{min}}\frac{N}{\sigma_t}$.
Writing $\log(1-x) = - x - \frac{x^2}{2} + R_3(x)$ with $\left| R_3(x) \right| \leq \frac{|x|^3}{3(1-\epsilon)^3}$ when $|x|<\epsilon<1$, we get
\begin{align}
\psi_i &\leq
\exp \Bigl(-\tau x + N \sum_{t=0}^{T-1}\Bigl( \frac{\sigma_t^2\tau^2}{2N^2} 
+  R_3 \Bigl(\frac{\sigma_t\tau}{N} \Bigr) \Bigr) \Bigr) \nonumber \\
&\leq \exp \Bigl(-N \Bigl( \frac{\tau x}{N} - \frac{\tau^2}{2N^2} \sum_{t=0}^{T-1} \sigma_t^2 \Bigr)  \Bigr) \exp  \Bigl( N \sum_{t=0}^{T-1} \Bigl| R_3 \Bigl(\frac{\sigma_t \tau}{N} \Bigr) \Bigr|  \Bigr) \label{expr}.
\end{align}  
We shall manage this expression by using Lemma~\ref{lm-B-Q}. In order to 
control the term $\exp(N \sum |R_3(\cdot)|)$, we make the choice
\begin{equation*}
\tau = \frac{axT}{\log T}
\end{equation*}
where $a$ is a parameter of order one to be optimized later.
From \eqref{sig_max} we get 
$\max_t\frac{\sigma_t \tau}{N} = O \left( (\log T)^{-1/2} \right)$. Hence, 
for all $T$ large, $\tau < \min_t \frac{N}{\sigma_t}$. Therefore, \eqref{pro} 
is valid for this choice of $\tau$ and for $T$ large. Moreover, for $\epsilon$ 
fixed and $T$ large, $\frac{\sigma_t \tau}{N} < \epsilon <1$ so that for 
these $T$
\begin{equation*}
N \sum_{t=0}^{T-1} \Bigl| R_3 \Bigl(\frac{\sigma_t\tau}{N} \Bigr) \Bigr| 
\leq \frac{a^3T^3x^3}{3N^2(1-\epsilon)^3 (\log T)^3} \sum_{t=0}^{T-1} |\sigma_t|^3 
= O \left( {T}{(\log T)^{-3/2}}\right) 
\end{equation*}
from (\ref{sum_sigma3}).
Plugging the expression of $\tau$ in (\ref{expr}), we get 
\begin{equation*}
\psi_i \leq \exp \Bigl(-N \Bigl( \frac{a T x^2}{(\log T)N} - \frac{a^2T^2x^2}{2N^2(\log T)^2} \sum_{t=0}^{T-1} \sigma_t^2 \Bigr)  \Bigr) \exp \Bigl( C \Bigl( {T}{(\log T)^{-3/2}} \Bigr) \Bigr) . 
\end{equation*}
Using (\ref{sum_sigma2}), we have
\begin{equation*}
\psi_i \leq \exp \Bigl(-\frac{x^2T}{\log T} 
\Bigl(a - \frac{\norme{\bs\Upsilon}_{\infty}^2a^2T}{N} \Bigr) \Bigr) 
\exp\Bigl( \frac{CT}{(\log T)^{3/2}} \Bigr). 
\end{equation*}
The right hand side term is minimized for $a=\frac{N}{2T\norme{\bs\Upsilon}_{\infty}^2}$ which finally gives
\begin{equation*}
\psi_i 
\leq \exp \Bigl(- \frac{Nx^2}{4\norme{\bs\Upsilon}_{\infty}^2\log T} (1 + o(1)) \Bigr).
\end{equation*}
Combining the above inequality with (\ref{union_bound}) (which induces additional $o(1)$ terms in the argument of the exponential) concludes the lemma.
\end{proof}

\begin{lemma}
\label{chi3-u} 
$\displaystyle{ 
\chi_3 \leq C T^{-\beta+2} \sqrt{\log T}
}$. 
\end{lemma}

\begin{proof}
From Lemma~\ref{lemma_d_quad2}, $\norme{R_T \odot B_T} \leq \norme{R_T}\norme{B_T}$ (see \cite[Theorem 5.5.1]{HornJoh'91}), and \eqref{lipschitz}, we get:
\begin{equation*}
\left| \E \widehat \Upsilon_T^u(\lambda_i) - \E \widehat \Upsilon_T^u(\lambda)  \right| 
\leq 2 \norme{d_T(\lambda) - d_T(\lambda_i)} \norme{R_T} \norme{B_T}
\leq C T^2 \left| \lambda - \lambda_{i} \right| \| \bs\Upsilon\|_\infty \sqrt{\log T}.
\end{equation*}
\end{proof}
Lemmas \ref{chi1-u}--\ref{chi3-u} show that $\PP[\chi_2 > x]$ dominates the
term $\PP[\chi_1 > x]$ and that the term $\chi_3$ is vanishing. Mimicking the
end of the proof of Theorem \ref{th-biased}, we obtain Theorem 
\ref{th-unbiased}. 

We conclude this section by an empirical evaluation by Monte Carlo simulations of $\mathbb{P}[\Vert {\widehat R_T - R_T}\Vert > x]$ (curves labeled Biased and Unbiased), with $\widehat R_T\in\{\widehat R_T^b,\widehat R_T^u\}$, $T=2N$, $x=2$. This is shown in Figure~\ref{det} against the theoretical exponential bounds of Theorems~\ref{th-biased} and \ref{th-unbiased} (curves labeled Biased theory and Unbiased theory). We observe that the rates obtained in Theorems~\ref{th-biased} and \ref{th-unbiased} are asymptotically close to optimal.

\begin{figure}[H]
\center
	\begin{tikzpicture}[font=\footnotesize]
		\renewcommand{\axisdefaulttryminticks}{2} 
		\pgfplotsset{every axis/.append style={mark options=solid, mark size=2pt}}
		\tikzstyle{every major grid}+=[style=densely dashed]       
			      %    \tikzstyle{every pin}=[fill=white,draw=black,font=\footnotesize,edge style={<-}] 
		\pgfplotsset{every axis legend/.append style={fill=white,cells={anchor=west},at={(0.99,0.02)},anchor=south east}}    
		\tikzstyle{every axis y label}+=[yshift=-10pt] 
		\tikzstyle{every axis x label}+=[yshift=5pt]
		\begin{axis}[
				grid=major,
				%ymajorgrids=false,
				xlabel={$N$},
				ylabel={$T^{-1} \log \left( \mathbb P \left[ \norme {\widehat R_T - R_T} > x \right] \right)$},
			    ytick={-0.2,-0.15,-0.1,-0.05,0},
                yticklabels = {$-0.2$,$-0.15$,$-0.1$,$-0.05$,$0$},
				xmin=10,
				xmax=40, 
				ymin=-0.2, 
				ymax=0,
      width=0.7\columnwidth,
      height=0.5\columnwidth
			]

      \addplot[smooth,black,densely dashed,line width=0.5pt,mark=star] plot coordinates{
(10.000000,-0.052766)(12.000000,-0.051276)(14.000000,-0.050322)(16.000000,-0.049673)(18.000000,-0.049212)(20.000000,-0.048872)(22.000000,-0.048614)(24.000000,-0.048414)(26.000000,-0.048255)(28.000000,-0.048127)(30.000000,-0.048023)(32.000000,-0.047936)(34.000000,-0.047864)(36.000000,-0.047803)(38.000000,-0.047750)(40.000000,-0.047705)

      };
      \addplot[smooth,black,line width=0.5pt,mark=star] plot coordinates{
(10.000000,-0.178687)(12.000000,-0.154996)(14.000000,-0.137000)(16.000000,-0.122875)(18.000000,-0.112712)(20.000000,-0.1042589)(22.000000,-0.096377)(24.000000,-0.091108)(26.000000,-0.085960)(28.000000,-0.082679)(30.000000,-0.080274)(32.000000,-0.077053)(34.000000,-0.074633)(36.000000,-0.071135)(38.000000,-0.069311)(40.000000,-0.067182)

      };
      \addplot[smooth,black,densely dashed,line width=0.5pt,mark=o] plot coordinates{
(10.000000,-0.011823)(12.000000,-0.010787)(14.000000,-0.010070)(16.000000,-0.009540)(18.000000,-0.009129)(20.000000,-0.008799)(22.000000,-0.008526)(24.000000,-0.008295)(26.000000,-0.008097)(28.000000,-0.007924)(30.000000,-0.007771)(32.000000,-0.007635)(34.000000,-0.007512)(36.000000,-0.007401)(38.000000,-0.007300)(40.000000,-0.007206)

      };
      
      \addplot[smooth,black,line width=0.5pt,mark=o] plot coordinates{
(10.000000,-0.063852)(12.000000,-0.050732)(14.000000,-0.041205)(16.000000,-0.034622)(18.000000,-0.029335)(20.000000,-0.025196)(22.000000,-0.021541)(24.000000,-0.018643)(26.000000,-0.015748)(28.000000,-0.013602)(30.000000,-0.012069)(32.000000,-0.010916)(34.000000,-0.010038)(36.000000,-0.007765)(38.000000,-0.007908)(40.000000,-0.007151)

      };

      \legend{ {Biased theory},{Biased},{Unbiased theory},{Unbiased}}
      \end{axis}
	\end{tikzpicture}
% \endpgfgraphicnamed 
\caption{Error probability of the spectral norm for $x=2$, $c=0.5$, $[R_T]_{k,l}=a^{|k-l|}$ with $a=0.6$.}
\label{det}
\end{figure}
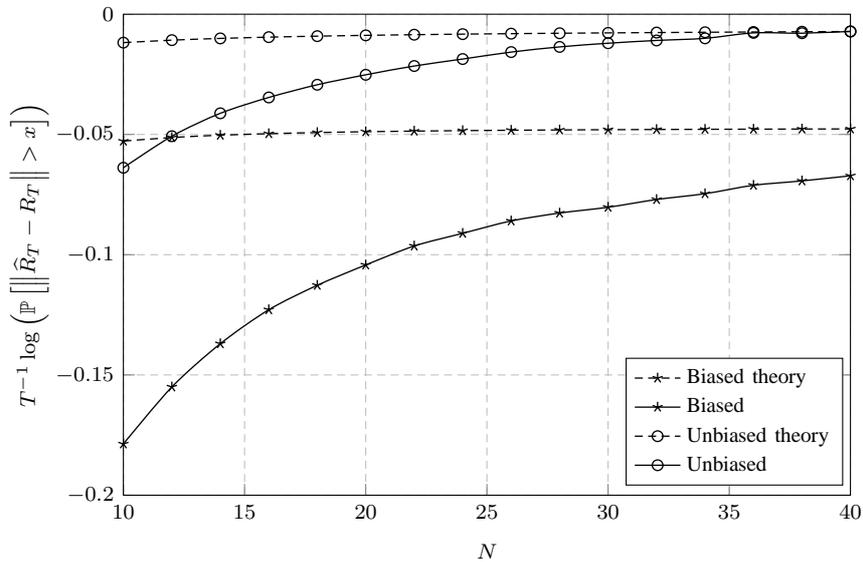

\section{Covariance matrix estimators for the \\ ``Signal plus Noise'' model} 
\label{sig+noise} 

\subsection{Model, assumptions, and results}
\label{subsec-model-perturbed} 

Consider now the following model:
\begin{equation}\label{model2}
	Y_T = [y_{n,t}]_{\substack{0\leq n\leq N-1 \\0 \leq t\leq T-1}} =P_T+V_T
\end{equation}
where the $N\times T$ matrix $V_T$ is defined in \eqref{model1} and where $P_T$ satisfies the 
following assumption: 
\begin{assumption} 
\label{ass-signal} 
$P_T \triangleq \bs h_T \bs s_T^{\sf H} \Gamma_T^{1/2}$ where $\bs h_T\in \C^N$ is a 
deterministic vector such that
$\sup_T \| \bs h_T \| < \infty$, the vector 
$\bs s_T = (s_0, \ldots, s_{T-1})^{\sf T} \in \C^T$ is a random vector
independent of $W_T$ with the distribution ${\cal CN}(0, I_T)$, and 
%$\Gamma_T^{1/2}$ is a nonnegative square root of $\Gamma_T$ with
$\Gamma_T = [\gamma_{ij} ]_{i,j=0}^{T-1}$ is Hermitian nonnegative such that $\sup_T \| \Gamma_T \| < \infty$. 
\end{assumption}

We have here a model for a rank-one signal corrupted with a Gaussian 
spatially white and temporally correlated noise with stationary temporal
correlations. Observe that the signal can also be temporally correlated. 
Our purpose is still to estimate the noise correlation matrix
$R_T$. To that end, we use one of the estimators \eqref{est-b} or \eqref{est-u} 
with the difference that the samples $v_{n,t}$ are simply replaced with the 
samples $y_{n,t}$. It turns out that these estimators are still consistent in 
spectral norm. Intuitively, $P_T$ does not break the 
consistence of these estimators as it can be seen as a rank-one 
perturbation of the noise term $V_T$ in which the subspace spanned by 
$(\Gamma^{1/2})^{\sf H} \bs s_T$ is ``delocalized'' enough so as not to perturb 
much the estimators of $R_T$. In fact, we even have the following strong result.
\begin{theorem}
\label{th-perturb} 
Let $Y_T$ be defined as in \eqref{model2} and let Assumptions~\ref{ass-rk}--\ref{ass-signal} hold. Define the
estimates 
\begin{align*}
\hat{r}_{k,T}^{bp}&= 
\frac{1}{NT} \sum_{n=0}^{N-1}\sum_{t=0}^{T-1}
y_{n,t+k} y_{n,t}^* \mathbbm{1}_{0 \leq t+k \leq T-1} \\
\hat{r}_{k,T}^{up}&= 
\frac{1}{N(T-|k|)} \sum_{n=0}^{N-1}\sum_{t=0}^{T-1}
y_{n,t+k} y_{n,t}^* \mathbbm{1}_{0 \leq t+k \leq T-1}
\end{align*}
and let 
\begin{align*}
\widehat R_T^{bp} &= {\cal T}(\hat{r}_{-(T-1),T}^{bp},\ldots,\hat{r}_{(T-1),T}^{bp}) \\
\widehat R_T^{up} &= {\cal T}(\hat{r}_{-(T-1),T}^{up},\ldots,\hat{r}_{(T-1),T}^{up}).
\end{align*}
Then for any $x > 0$, 
\begin{equation*} 
\mathbb{P} \left[\norme{\widehat R_T^{bp} - R_T} > x \right]
\leq 
\exp \Bigl( -cT \Bigl( \frac{x}{\| \bs\Upsilon \|_\infty} - 
 \log \Bigl( 1 + \frac{x}{\| \bs\Upsilon \|_\infty} \Bigr) + o(1) \Bigr) 
\Bigr) 
\end{equation*} 
and 
\begin{equation*} 
\mathbb{P} \left[\norme{\widehat R_T^{up} - R_T} > x \right]
\leq 
\exp \Bigl(- \frac{cTx^2}{4\norme{\bs\Upsilon}_{\infty}^2\log T} (1 + o(1)) 
\Bigr).  
\end{equation*} 
\end{theorem}

Before proving this theorem, some remarks are in order.
\begin{remark}
Theorem \ref{th-perturb} generalizes without difficulty to the case where $P_T$ has a fixed rank $K>1$. This captures the situation of $K\ll \min(N,T)$ sources.
\end{remark} 
\begin{remark}
	Similar to the proofs of Theorems \ref{th-biased} and \ref{th-unbiased}, the proof of Theorem~\ref{th-perturb} uses concentration inequalities for functionals of Gaussian random variables based on the moment generating function and the Chernoff bound. Exploiting instead McDiarmid's concentration inequality \cite{ledoux}, it is possible to adapt Theorem~\ref{th-perturb} to $\bs s_T$ with bounded (instead of Gaussian) entries. This adaptation may account for discrete sources met in digital communication signals. 
\end{remark} 

\subsection{Main elements of the proof of Theorem \ref{th-perturb}} 

We restrict the proof to the more technical part that concerns $\widehat R^{up}_T$. 
Defining
\begin{eqnarray*}
\widehat \Upsilon_T^{up}(\lambda) \triangleq \sum_{k=-(T-1)}^{T-1} \hat{r}_{k,T}^{up} e^{ik \lambda}
\end{eqnarray*}
and recalling that $\Upsilon_T(\lambda) = \sum_{k=-(T-1)}^{T-1} r_{k} e^{ik \lambda}$, 
we need to establish a concentration inequality on \linebreak
$\PP \left[ \sup_{\lambda\in[0,2\pi)} | \widehat \Upsilon_T^{up}(\lambda) - 
\Upsilon_T(\lambda) | > x \right]$. For any $\lambda\in[0,2\pi)$, the term 
$\widehat \Upsilon_T^{up}(\lambda)$ can be written as 
(see Lemma~\ref{lemma_d_quad2})  
\begin{align*}
\widehat \Upsilon_T^{up}(\lambda)&= 
d_T(\lambda)^{\sf H} \left(\frac{Y_T^{\sf H}Y_T}{N} \odot B_T \right)
d_T(\lambda)  \\
&= 
d_T(\lambda)^{\sf H} \left(\frac{V_T^{\sf H}V_T}{N} \odot B_T \right)
d_T(\lambda) \\
&+ 
d_T(\lambda)^{\sf H} \left(\frac{P_T^{\sf H}V_T+V_T^{\sf H}P_T}{N} \odot B_T 
\right) d_T(\lambda) \\
&+  d_T(\lambda)^{\sf H} \left(\frac{P_T^{\sf H}P_T}{N} \odot B_T \right)
d_T(\lambda)  \\
&\triangleq \widehat \Upsilon_T^{u}(\lambda) + 
\widehat \Upsilon_T^{cross}(\lambda) + \widehat \Upsilon_T^{sig}(\lambda) 
\end{align*}  
where $B_T$ is the matrix defined in the statement of 
Lemma~\ref{lemma_d_quad2}.
We know from the proof of Theorem~\ref{th-unbiased} that 
\begin{equation} 
\label{noise-term} 
\mathbb{P} \left[\sup_{\lambda\in[0,2\pi)} 
| \widehat \Upsilon_T^{u}(\lambda) - \Upsilon_T(\lambda) | > x \right]
\leq 
\exp \left(- \frac{cTx^2}{4\norme{\bs\Upsilon}_{\infty}^2\log T} (1 + o(1)) 
\right) .  
\end{equation} 
We then need only handle the terms $\widehat \Upsilon_T^{cross}(\lambda)$ and 
$\widehat \Upsilon_T^{sig}(\lambda)$. 

We start with a simple lemma. 
\begin{lemma}
\label{lm-prod-gauss}
Let $X$ and $Y$ be two independent ${\cal N}(0,1)$ random variables. Then for 
any $\tau \in(-1,1)$, 
\[
\E[\exp(\tau XY)] = (1 - \tau^2)^{-1/2}.
\]
\end{lemma}
\begin{proof}
\begin{align*}
\E[\exp(\tau XY)] &= \frac{1}{2\pi} \int_{\R^2} 
e^{\tau xy} e^{-x^2/2} e^{-y^2/2} \, dx\, dy \\
&= \frac{1}{2\pi} \int_{\R^2} 
e^{-(x-\tau y)^2/2} e^{-(1 - \tau^2)y^2/2} \, dx\, dy \\
&= (1 - \tau^2)^{-1/2} .
\end{align*}
\end{proof} 
With this result, we now have
\begin{lemma} 
\label{cross-term}
There exists a constant $a > 0$ such that 
\[
\PP\left[ \sup_{\lambda \in[0,2\pi)} | \widehat \Upsilon_T^{cross}(\lambda) | 
> x \right] \leq \exp\Bigl( - \frac{axT}{\sqrt{\log T}}(1 + o(1)) \Bigr) . 
\]
\end{lemma}
\begin{proof}
We only sketch the proof of this lemma. We show that for any 
$\lambda \in [0, 2\pi]$, 
\[
\PP[ | \widehat \Upsilon_T^{cross}(\lambda) | 
> x ] \leq \exp\Bigl( - \frac{axT}{\sqrt{\log T}} + C \Bigr) 
\]
where $C$ does not depend on $\lambda \in [0,2\pi]$. The lemma is then proven by a discretization
argument of the interval $[0, 2\pi]$ analogous to what was done in the 
proofs of Section~\ref{unperturbed}.
We shall bound $\PP[ \widehat \Upsilon_T^{cross}(\lambda) > x ]$, the term
$\PP[ \widehat \Upsilon_T^{cross}(\lambda) < - x ]$ being bounded similarly. 
From Lemma~\ref{lemma_hadamard}, we get 
\begin{align*}
\widehat \Upsilon_T^{cross}(\lambda) &= 
\tr \Bigl( D_T(\lambda)^{\sf H} \frac{P^{\sf H}_T V_T + V_T^{\sf H} P_T}{N} 
D_T(\lambda) B_T \Bigr) \\
&= \tr \frac{D_T(\lambda)^{\sf H} (\Gamma_T^{1/2})^{\sf H} {\bs s}_T 
{\bs h}_T^{\sf H} W_T R_T^{1/2} D_T(\lambda) B_T}{N} \\
&\phantom{=} + 
\tr \frac{D_T(\lambda)^{\sf H} (R_T^{1/2})^{\sf H} W_T^{\sf H} {\bs h}_T 
{\bs s}_T^{\sf H} \Gamma_T^{1/2} D_T(\lambda) B_T}{N} \\
&= \frac 2N \Re ( \bs h_T^{\sf H} W_T G_T(\lambda) \bs s_T ) 
\end{align*} 
where $G_T(\lambda) = R_T^{1/2} D_T(\lambda) B_T D_T(\lambda)^{\sf H} 
(\Gamma_T^{1/2})^{\sf H}$. Let $G_T(\lambda) = U_T \Omega_T 
\widetilde U_T^{\sf H}$ be a singular value decomposition of $G_T(\lambda)$ where 
$\Omega = \diag(\omega_0, \ldots, \omega_{T-1})$. Observe that the vector 
$\bs x_T \triangleq W_T^{\sf H} \bs h_T = (x_0,\ldots, x_{T-1})^{\sf T}$ has 
the distribution ${\cal CN}(0, \| \bs h_T \|^2 I_T)$. We can then write 
\[
\widehat \Upsilon_T^{cross}(\lambda) \stackrel{{\cal L}}{=} 
\frac 2N \Re\left( \bs x_T^{\sf H} \Omega_T \bs s_T \right)
= \frac 2N \sum_{t=0}^{T-1} \omega_t ( \Re x_t \Re s_t + \Im x_t \Im s_t). 
\]
Notice that $\{ \Re x_t, \Im x_t, \Re s_t, \Im s_t \}_{t=0}^{T-1}$ are
independent with $\Re x_t, \Im x_t \sim {\cal N}(0, \| \bs h_T \|^2/2)$ and 
$\Re s_t, \Im s_t \sim {\cal N}(0, 1/2)$. Letting $0 < \tau < (\sup_T \| \bs h_T \|)^{-1}(\sup_{\lambda} \| G_T(\lambda) \|)^{-1}$ and using Markov's inequality and Lemma~\ref{lm-prod-gauss}, we get 
\begin{align*} 
\PP \left[ \widehat \Upsilon_T^{cross}(\lambda) > x \right] &= 
 \PP \left[ e^{N \tau \widehat \Upsilon_T^{cross}(\lambda)} > e^{N \tau x} \right] 
\leq e^{-N \tau x} \E \left[ e^{2\tau \sum_t 
\omega_t ( \Re x_t \Re s_t + \Im x_t \Im s_t)} \right] \\
&= e^{-N\tau x} \prod_{t=0}^{T-1} \left( 1 - \tau^2 \omega_t^2 \| \bs h_T \|^2 \right)^{-1} 
= \exp \left( -N \tau x - \sum_{t=0}^{T-1} \log( 1 - \tau^2 \omega_t^2 \| \bs h_T\|^2 ) \right). 
\end{align*}
Mimicking the proof of Lemma~\ref{lm-B-Q}, we can establish that 
$\sum_t \omega_t^2 = O(\log T)$ and $\max_t \omega_t = O(\sqrt{\log T})$ 
uniformly in $\lambda \in [0, 2\pi]$. Set $\tau = b / \sqrt{\log T}$ where
$b > 0$ is small enough so that 
$\sup_{T,\lambda} (\tau \| \bs h_T \| \, \| G_T(\lambda) \|) < 1$. 
Observing that $\log(1-x) = O(x)$ for $x$ small enough, we get 
\[
\PP[ \widehat \Upsilon_T^{cross}(\lambda) > x ] \leq 
\exp \bigl( -N bx/\sqrt{\log T} + {\cal E}(\lambda, T) \bigr)  
\]
where $| {\cal E}(\lambda, T) | \leq (C / \log T) \sum_t \omega_t^2 \leq C$. 
This establishes Lemma~\ref{cross-term}. 
\end{proof}

\begin{lemma} 
\label{signal-term}
There exists a constant $a > 0$ such that 
\[
\PP\left[ \sup_{\lambda \in[0,2\pi)} | \widehat \Upsilon_T^{sig}(\lambda) | 
> x \right] \leq \exp\Bigl( - \frac{axT}{\sqrt{\log T}}(1 + o(1)) \Bigr) . 
\]
\end{lemma}
\begin{proof}
By Lemma~\ref{lemma_hadamard},
\begin{align*} 
\widehat \Upsilon_T^{sig}(\lambda) &= 
N^{-1} \tr ( D_T^{\sf H} P_T^{\sf H} P_T D_T B_T ) \\
&= \frac{\| \bs h_T \|^2}{N} \bs s_T^{\sf H} G_T(\lambda) \bs s_T 
\end{align*} 
where $G_T(\lambda) = \Gamma_T^{1/2} D_T(\lambda) B_T D_T(\lambda)^{\sf H} 
(\Gamma_T^{1/2})^{\sf H}$. By the spectral factorization 
$G_T(\lambda) = U_T \Sigma_T U_T^{\sf H}$ with 
$\Sigma_T = \diag(\sigma_0, \ldots, \sigma_{T-1})$, we get
\[
\widehat \Upsilon_T^{sig}(\lambda) \stackrel{{\cal L}}{=} 
\frac{\| \bs h_T \|^2}{N} \sum_{t=0}^{T-1} \sigma_t | s_t |^2 
\] 
and
\begin{align*} 
\PP[ \widehat \Upsilon_T^{sig}(\lambda) > x ] &\leq 
e^{-N\tau x} \E \Bigl[ e^{\tau \|\bs h_T\|^2 \sum_t \sigma_t | s_t|^2}\Bigr] \\
&= \exp\Bigl( -N \tau x - \sum_{t=0}^{T-1} 
\log(1 - \sigma_t \tau \|\bs h_T\|^2) \Bigr) 
\end{align*} 
for any $\tau \in (0, 1 / (\|\bs h_T\|^2 \sup_\lambda \| G_T(\lambda)\|))$. 
Let us show that 
\[
| \tr G_T(\lambda) | \leq C \sqrt{\frac{\log T + 1}{T}} . 
\] 
Indeed, we have 
\begin{align*}
| \tr G_T(\lambda) | &= N^{-1} | \tr D_T B_T D_T^{\sf H} \Gamma_T | = \frac 1N \left| \sum_{k,\ell=0}^{T-1} \frac{e^{-\imath (k-\ell)\lambda} 
\gamma_{\ell,k}}{T-|k-\ell|} \right| \\ 
&\leq \Bigl( \frac 1N \sum_{k,\ell=0}^{T-1} |\gamma_{k,\ell}|^2 \Bigr)^{1/2} 
\Bigl( \frac 1N \sum_{k,\ell=0}^{T-1} \frac{1}{(T-|k-\ell|)^2} \Bigr)^{1/2} \\
&= \Bigl(\frac{\tr \Gamma_T\Gamma_T^{\sf H}}{N} \Bigr)^{1/2} 
\Bigl(\frac 2N (\log T + C) \Bigr)^{1/2} \leq C \sqrt{\frac{\log T + 1}{T}} .
\end{align*} 
Moreover, similar to the proof of Lemma~\ref{lm-B-Q}, we can show that 
$\sum_t\sigma_t^2 = O(\log T)$ and $\max_t |\sigma_t| = O(\sqrt{\log T})$ 
uniformly in $\lambda$. Taking $\tau = b / \sqrt{\log T}$ for 
$b > 0$ small enough, and recalling that $\log(1-x) = 1 - x + O(x^2)$
for $x$ small enough, we get that 
\[
\PP[ \widehat \Upsilon_T^{sig}(\lambda) > x ] \leq 
\exp\Bigl( - \frac{N bx}{\sqrt{\log T}} + 
\frac{b \| \bs h_T \|^2}{\sqrt{\log T}} \tr G_T(\lambda) + 
{\cal E}(T,\lambda) \Bigr) 
\]
where $| {\cal E}(T,\lambda) | \leq (C / \log T) \sum_t \sigma_t^2 \leq C$. 
We therefore get 
\[
\PP[ \widehat \Upsilon_T^{sig}(\lambda) > x ] \leq 
\exp\Bigl( - \frac{N bx}{\sqrt{\log T}} + C \Bigr) 
\]
where $C$ is independent of $\lambda$. Lemma~\ref{signal-term} is then obtained
by the discretization argument of the interval $[0,2\pi]$.
\end{proof}

Gathering Inequality~\eqref{noise-term} with Lemmas~\ref{cross-term} and 
\ref{signal-term}, we get the second inequality of the statement of 
Theorem~\ref{th-perturb}.

\section{Application to source detection} \label{detect}
%. From an application point of view, this model corresponds to a general transmission model and the matrix $P_T$ represents the signal matrix and $V_T=W_TR_T^{1/2}$ corresponds to a temporally correlated noise. The perturbation matrix 

Consider a sensor network composed of $N$ sensors impinged by zero (hypothesis $H_0$) or one (hypothesis $H_1$) source signal. The stacked signal matrix $Y_T=[y_0,\ldots,y_{T-1}]\in\C^{N\times T}$ from time $t=0$ to $t=T-1$ is modeled as
\begin{eqnarray}\label{model_det}
Y_T = \left\{
    \begin{array}{ll}
        V_T & \mbox{, $H_0$} \\
        \bs h_T \bs s_T^{\sf H} + V_T& \mbox{, $H_1$}
    \end{array}
\right.
\end{eqnarray}
where $s_T^{\sf H}=[s_0^*,\ldots,s_{T-1}^*]$ are (hypothetical) independent $\mathcal{CN}(0,1)$ signals transmitted through the constant channel $\bs h_T \in \C^{N}$, and $V_T=W_T R_T^{1/2}\in\C^{N\times T}$ models a stationary noise matrix as in \eqref{model1}.

As opposed to standard procedures where preliminary pure noise data are available , we shall proceed here to an online signal detection test solely based on $Y_T$, by exploiting the consistence established in Theorem~\ref{th-perturb}.
The approach consists precisely in estimating $R_T$ by $\widehat{R}_T\in\{\widehat{R}_T^{bp},\widehat{R}_T^{up}\}$, which is then used as a whitening matrix for $Y_T$. The binary hypothesis \eqref{model_det} can then be equivalently written
\begin{eqnarray}\label{model_w}
Y_T \widehat{R}_T^{-1/2}= \left\{
    \begin{array}{ll}
        W_T {R}_T \widehat{R}_T^{-1/2}& \mbox{, $H_0$}  \\
        \bs h_T \bs s_T^{\sf H} \widehat{R}_T^{-1/2} + W_T {R}_T \widehat{R}_T^{-1/2}& \mbox{, $H_1$}.
    \end{array}
\right.
\end{eqnarray}
Since $\Vert R_T\widehat{R}_T^{-1}-I_T\Vert\to 0$ almost surely (by Theorem~\ref{th-perturb} as long as $\inf_{\lambda\in[0,2\pi)} {\bs \Upsilon}(\lambda)>0$), for $T$ large, the decision on the hypotheses \eqref{model_w} can be handled by the generalized likelihood ratio test (GLRT) \cite{BianDebMai'11} by approximating $W_T {R}_T \widehat{R}_T^{-1/2}$ as a purely white noise. We then have the following result.
\begin{theorem}\label{th_detection}
	Let $\widehat{R}_T$ be any of $\widehat{R}_T^{bp}$ or $\widehat{R}_T^{up}$ strictly defined in Theorem~\ref{th-perturb} for $Y_T$ now following model \eqref{model_det}. Further assume $\inf_{\lambda\in[0,2\pi)} {\bs \Upsilon}(\lambda)>0$ and define the test
\begin{equation}\label{glrt_est}
\alpha = \frac{N\norme{Y_T \widehat{R}_T^{-1}Y_T^{\sf H}}}{\tr \left( Y_T \widehat{R}_T^{-1} Y_T^{\sf H} \right)} ~ \overset{H_0}{\underset{H_1}{\lessgtr}} ~ \gamma
\end{equation}
where $\gamma\in\R^+$ satisfies $\gamma>(1+\sqrt c)^2$. Then, as $T \to \infty$,
\begin{align*}
	\mathbb{P} \left[ \alpha \geq \gamma \right] \to \left\{ \begin{array}{ll} 0 &,~H_0 \\ 1 &,~H_1. \end{array}\right.
\end{align*}
\end{theorem}

Recall from \cite{BianDebMai'11} that the decision threshold $(1+\sqrt{c})^2$ corresponds to the almost sure limiting largest eigenvalue of $\frac1T W_T W_T^{\sf H}$, that is the right-edge of the support of the Mar\v{c}enko--Pastur law. 

Simulations are performed hereafter to assess the performance of the test \eqref{glrt_est} under several system settings. We take here ${\bs h}_T$ to be the following steering vector ${\bs h}_T=\sqrt{p/T}[1, \ldots , e^{2i\pi \theta (T-1)}]$ with $\theta = 10^\circ$ and $p$ a power parameter. The matrix $R_T$ models an autoregressive process of order 1 with parameter $a$, {\it i.e.} $[R_T]_{k,l}=a^{|k-l|}$. 

In Figure~\ref{det1}, the detection error $1-\mathbb{P} [ \alpha \geq \gamma|H_1]$ of the test \eqref{glrt_est} for a false alarm rate (FAR) $\mathbb{P} [ \alpha \geq \gamma|H_0 ]=0.05$ under $\widehat{R}_T=\widehat{R}_T^{up}$ (Unbiased) or $\widehat{R}_T=\widehat{R}_T^{bp}$ (Biased) is compared against the estimator that assumes $R_T$ perfectly known (Oracle), {\it i.e.} that sets $\widehat{R}_T=R_T$ in \eqref{glrt_est}, and against the GLRT test that wrongly assumes temporally white noise (White), {\it i.e.} that sets $\widehat{R}_T=I_T$ in \eqref{glrt_est}. The source signal power is set to $p=1$, that is a signal-to-noise ratio (SNR) of $0$ dB, $N$ is varied from $10$ to $50$ and $T=N/c$ for $c=0.5$ fixed. In the same setting as Figure~\ref{det1}, the number of sensors is now fixed to $N=20$, $T=N/c=40$ and the SNR (hence $p$) is varied from $-10$~dB to $4$~dB. The powers of the various tests are displayed in Figure~\ref{det2} and compared to the detection methods which estimate $R_T$ from a pure noise sequence called Biased PN (pure noise) and Unbiased PN. The results of the proposed online method are close to that of Biase/Unbiased PN, this last presenting the disadvantage to have at its disposal a pure noise sequence at the receiver.  

Both figures suggest a close match in performance between Oracle and Biased, while Unbiased shows weaker performance. The gap evidenced between Biased and Unbiased confirms the theoretical conclusions. 

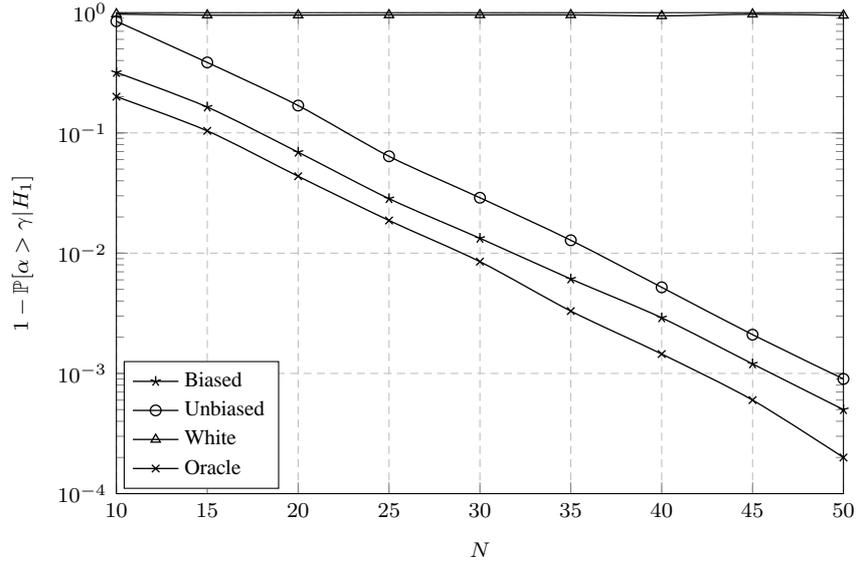
\begin{figure}[H]
\center
  \begin{tikzpicture}[font=\footnotesize]
    \renewcommand{\axisdefaulttryminticks}{2} 
    \pgfplotsset{every axis/.append style={mark options=solid, mark size=2pt}}
    \tikzstyle{every major grid}+=[style=densely dashed]       
%    \tikzstyle{every pin}=[fill=white,draw=black,font=\footnotesize,edge style={<-}] 
\pgfplotsset{every axis legend/.append style={fill=white,cells={anchor=west},at={(0.01,0.01)},anchor=south west}}    
    \tikzstyle{every axis y label}+=[yshift=-10pt] 
    \tikzstyle{every axis x label}+=[yshift=5pt]
    \begin{semilogyaxis}[
      grid=major,
      %ymajorgrids=false,
      xlabel={$N$},
      ylabel={$1-\mathbb{P}[\alpha>\gamma|H_1]$},
      xmin=10,
      xmax=50, 
      ymin=1e-4, 
      ymax=1,
      width=0.7\columnwidth,
      height=0.5\columnwidth
      ]
      \addplot[smooth,black,line width=0.5pt,mark=star] plot coordinates{
(10.000000,0.318300)(15.000000,0.164000)(20.000000,0.069100)(25.000000,0.028400)(30.000000,0.013300)(35.000000,0.006100)(40.000000,0.002900)(45.000000,0.001200)(50.000000,0.0005000)

      };
      
      \addplot[smooth,black,line width=0.5pt,mark=o] plot coordinates{
(10.000000,0.848100)(15.000000,0.385600)(20.000000,0.168900)(25.000000,0.063900)(30.000000,0.028900)(35.000000,0.012800)(40.000000,0.005200)(45.000000,0.002100)(50.000000,0.000900)

      };
      
      \addplot[smooth,black,line width=0.5pt,mark=triangle] plot coordinates{
(10.000000,0.973000)(15.000000,0.953000)(20.000000,0.954000)(25.000000,0.957000)(30.000000,0.958000)(35.000000,0.958000)(40.000000,0.940000)(45.000000,0.970000)(50.000000,0.949000)

      };
      
      \addplot[smooth,black,line width=0.5pt,mark=x] plot coordinates{
(10.000000,0.200200)(15.000000,0.104000)(20.000000,0.043600)(25.000000,0.018700)(30.000000,0.008500)(35.000000,0.003300)(40.000000,0.0014500)(45.000000,0.000600)(50.000000,0.000200)

      };

	\legend{{Biased},{Unbiased},{White},{Oracle}}
       \end{semilogyaxis}
       \end{tikzpicture}
       \caption{Detection error versus $N$ with FAR$=0.05$, $p=1$, SNR$=0$ dB, $c=0.5$, and $a=0.6$.}
\label{det1}
\end{figure}

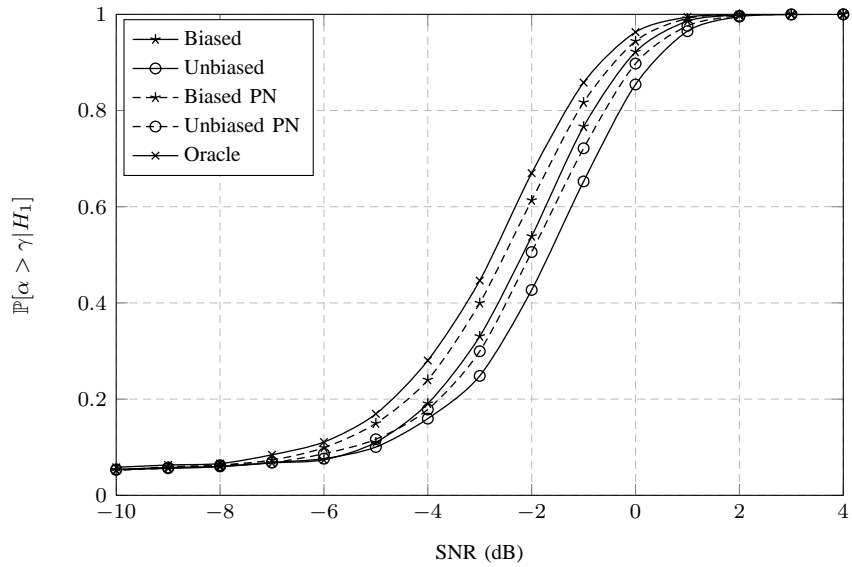
\begin{figure}[H]
\center
  \begin{tikzpicture}[font=\footnotesize]
    \renewcommand{\axisdefaulttryminticks}{2} 
    \pgfplotsset{every axis/.append style={mark options=solid, mark size=2pt}}
    \tikzstyle{every major grid}+=[style=densely dashed]       
%    \tikzstyle{every pin}=[fill=white,draw=black,font=\footnotesize,edge style={<-}] 
\pgfplotsset{every axis legend/.append style={fill=white,cells={anchor=west},at={(0.01,0.99)},anchor=north west}}    
    \tikzstyle{every axis y label}+=[yshift=-10pt] 
    \tikzstyle{every axis x label}+=[yshift=5pt]
    \begin{axis}[
      grid=major,
      %ymajorgrids=false,
      xlabel={SNR (dB)},
      ylabel={$\mathbb{P}[\alpha>\gamma|H_1]$},
      xmin=-10,
      xmax=4, 
      ymin=0, 
      ymax=1,
      width=0.7\columnwidth,
      height=0.5\columnwidth
      ]
      
            \addplot[smooth,black,line width=0.5pt,mark=star] plot coordinates{
(-15.000000,0.051200)(-14.000000,0.048100)(-13.000000,0.049100)(-12.000000,0.053100)(-11.000000,0.049800)(-10.000000,0.053200)(-9.000000,0.056100)(-8.000000,0.0593000)(-7.000000,0.067500)(-6.000000,0.073900)(-5.000000,0.110200)(-4.000000,0.190800)(-3.000000,0.330600)(-2.000000,0.538500)(-1.000000,0.766600)(0.000000,0.922000)(1.000000,0.985800)(2.000000,0.998600)(3.000000,1.000000)(4.000000,1.000000)(5.000000,1.000000)

      };
      
      \addplot[smooth,black,line width=0.5pt,mark=o] plot coordinates{
(-15.000000,0.053900)(-14.000000,0.049300)(-13.000000,0.049500)(-12.000000,0.055100)(-11.000000,0.050800)(-10.000000,0.054000)(-9.000000,0.057800)(-8.000000,0.060100)(-7.000000,0.068400)(-6.000000,0.076300)(-5.000000,0.100400)(-4.000000,0.159500)(-3.000000,0.248400)(-2.000000,0.427000)(-1.000000,0.652600)(0.000000,0.854200)(1.000000,0.964900)(2.000000,0.995100)(3.000000,0.999500)(4.000000,1.000000)(5.000000,1.000000)

      };
      
      \addplot[smooth,black,densely dashed,line width=0.5pt,mark=star] plot coordinates{
(-15.000000,0.047700)(-14.000000,0.047400)(-13.000000,0.044600)(-12.000000,0.049700)(-11.000000,0.049700)(-10.000000,0.051900)(-9.000000,0.059400)(-8.000000,0.063100)(-7.000000,0.073300)(-6.000000,0.099100)(-5.000000,0.149500)(-4.000000,0.240100)(-3.000000,0.399900)(-2.000000,0.613300)(-1.000000,0.816900)(0.000000,0.944300)(1.000000,0.991500)(2.000000,0.999200)(3.000000,1.000000)(4.000000,1.000000)(5.000000,1.000000)

      };
      
      \addplot[smooth,black,densely dashed,line width=0.5pt,mark=o] plot coordinates{
(-15.000000,0.048300)(-14.000000,0.051100)(-13.000000,0.047100)(-12.000000,0.050300)(-11.000000,0.051800)(-10.000000,0.052900)(-9.000000,0.056300)(-8.000000,0.062600)(-7.000000,0.068100)(-6.000000,0.086500)(-5.000000,0.116600)(-4.000000,0.178600)(-3.000000,0.299600)(-2.000000,0.505700)(-1.000000,0.721500)(0.000000,0.897600)(1.000000,0.976000)(2.000000,0.996900)(3.000000,0.999900)(4.000000,0.999900)(5.000000,1.000000)

      };

      \addplot[smooth,black,line width=0.5pt,mark=x] plot coordinates{
(-15.000000,0.048800)(-14.000000,0.050800)(-13.000000,0.048300)(-12.000000,0.048900)(-11.000000,0.055600)(-10.000000,0.058300)(-9.000000,0.063100)(-8.000000,0.066000)(-7.000000,0.084300)(-6.000000,0.110800)(-5.000000,0.169400)(-4.000000,0.280300)(-3.000000,0.446500)(-2.000000,0.669800)(-1.000000,0.858100)(0.000000,0.963000)(1.000000,0.994300)(2.000000,0.999400)(3.000000,1.000000)(4.000000,1.000000)(5.000000,1.000000)

      };

	\legend{{Biased},{Unbiased},{Biased PN},{Unbiased PN},{Oracle}}
       \end{axis}
       \end{tikzpicture}
       \caption{Power of detection tests versus SNR (dB) with FAR$=0.05$, $N=20$, $c=0.5$, and $a=0.6$.}
\label{det2}
\end{figure}

\begin{appendix}
\subsection{Proofs for Theorem~\ref{th-biased}} 
\subsubsection{Proof of Lemma \ref{lemma_d_quad}} 
\label{anx-lm-qf} 

Developing the quadratic forms given in the statement of the lemma, we get 
\begin{align*} 
d_T(\lambda)^{\sf H}\frac{V_T^{\sf H}V_T}{N}d_T(\lambda) &= 
\frac{1}{NT}\sum_{l,l'=0}^{T-1} e^{-\imath(l'-l)\lambda} [V_T^{\sf H}V_T]_{l,l'} \\
&= \frac{1}{NT}\sum_{l,l'=0}^{T-1} e^{-\imath(l'-l)\lambda} 
\sum_{n=0}^{N-1} v^*_{n,l} v_{n,l'} \\ 
&= \sum_{k=-(T-1)}^{T-1} e^{-\imath k \lambda} 
\frac{1}{NT} \sum_{n=0}^{N-1} \sum_{t=0}^{T-1} v_{n,t}^* v_{n,t+k} 
\mathbbm{1}_{0 \leq t+k \leq T-1}\\ 
&= \sum_{k=-(T-1)}^{T-1} \hat{r}_k^b e^{-\imath k \lambda}=
\widehat\Upsilon_T^b(\lambda), 
\end{align*} 
and 
\begin{align*} 
\E \left[ d_T(\lambda)^{\sf H} \frac{V_T^{\sf H}V_T}{N}d_T(\lambda) \right] 
&= d_T(\lambda)^{\sf H} (R_T^{1/2})^{\sf H} 
\frac{\E[ W_T^{\sf H} W_T]}{N} R_T^{1/2} d_T(\lambda) \\ 
&= d_T(\lambda)^{\sf H} R_T d_T(\lambda) .
\end{align*} 

%\section{Proofs for Theorem~\ref{th-unbiased}}
\subsection{Proofs for Theorem~\ref{th-unbiased}}
\subsubsection{Proof of Lemma \ref{lemma_d_quad2}} 
\label{anx-lm-qf2} 
We have 
\begin{align*}
d_T(\lambda)^{\sf H} \left( \frac{V_T^{\sf H}V_T}{N} \odot B_T \right) d_T(\lambda) &=
\frac{1}{NT}\sum_{l,l'=-(T-1)}^{T-1}e^{i(l-l') \lambda} [V_T^{\sf H}V_T]_{l,l'}\frac{T}{T-|l-l'|} \nonumber \\
&= \sum_{k=-(T-1)}^{T-1}e^{ik\lambda}\frac{1}{N(T-|k|)}\sum_{n=0}^{N-1}\sum_{t=0}^{T-1}v_{n,t}^*v_{n,t+k}\mathbbm{1}_{0 \leq t+k \leq T-1} \\
&=\sum_{k=-(T-1)}^{T-1} \hat{r}_k^u e^{ik \lambda} =\widehat \Upsilon_T^u(\lambda).
\end{align*}

\subsubsection{Proof of Lemma \ref{lm-B-Q}} 
\label{prf-lm-B-Q} 
We start by observing that 
\begin{align*}
\tr B_T^2 &= \sum_{i,j=0}^{T-1} \left[ B_T \right]_{i,j}^2 
= \sum_{i,j=0}^{T-1} \left( \frac{T}{T-|i-j|} \right)^2 
= 2\sum_{i>j}^{T-1} \left( \frac{T}{T-|i-j|} \right)^2 + T \\
&= 2 \sum_{k=1}^{T-1} \left( \frac{T}{T-k} \right)^2 \left( T - k \right) + T 
= 2 T^2 \sum_{k=1}^{T-1} \frac{1}{T-k} + T 
= 2 T^2 \left(\log T + C \right).
\end{align*}
Inequality \eqref{norme_B} is then obtained upon noticing that 
$\norme{B_T} \leq \sqrt{\tr B_T^2}$. 

We now show (\ref{sum_sigma2}). 
Using twice the inequality $\tr (FG) \leq \norme{F} \tr(G)$ when 
$F,G \in \C^{m \times m}$ and $G$ is nonnegative definite \cite{HornJoh'91}, 
we get
\begin{align*}
\sum_{t=0}^{T-1} \sigma_t^2(\lambda_i) &= \tr Q_T(\lambda_i)^2  
=  \tr  R_T D_T(\lambda_i)B_T D_T(\lambda_i)^{\sf H} R_T D_T(\lambda_i) 
B_T D_T(\lambda_i)^{\sf H} \\
&\leq \norme{R_T} \tr R_T (D_T(\lambda_i) B_T D_T(\lambda_i)^{\sf H})^2 \\
&\leq T^{-2} \norme{R_T}^2 \tr (B_T^2) \leq 2\norme{\bs\Upsilon}_{\infty}^2 \log T + C. 
\end{align*}

Inequality \eqref{sig_max} is immediate since $\norme{Q_T}^2 \leq \tr Q_T^2$.

As regards \eqref{sum_sigma3}, by the Cauchy--Schwarz inequality,
\begin{align*}
\sum_{t=0}^{T-1} |\sigma_t^3(\lambda_i)| &= 
\sum_{t=0}^{T-1} \sigma_t^2(\lambda_i) |\sigma_t(\lambda_i)|  
\leq \sqrt{\sum_{t=0}^{T-1} \sigma_t^4(\lambda_i) \sum_{t=0}^{T-1} \sigma_t^2(\lambda_i)} \\ 
&\leq \sqrt{\left(\sum_{t=0}^{T-1} \sigma_t^2(\lambda_i) \right)^2 \sum_{t=0}^{T-1} \sigma_t^2(\lambda_i)} = \left( \sum_{t=0}^{T-1} \sigma_t^2(\lambda_i) \right)^{3/2}\\
&= C ( (\log T)^{3/2} +1 ).
\end{align*}

% \subsection{Proofs for Theorem~\ref{th-biased-dp}}
% \subsubsection{Proof of \eqref{Upsilon_up} and \eqref{EUpsilon_up}}\label{anx-lm-qf3}
% Similarly to the proof of Lemma~\ref{lemma_d_quad}, we have:
% \begin{align*} 
% d_T(\lambda)^{\sf H} \left( \frac{Y_T^{\sf H}Y_T}{N} \odot B_T \right)d_T(\lambda)  
% &= \sum_{k=-(T-1)}^{T-1} \widehat{r}_k^p e^{-\imath k \lambda} \\&=
% \widehat \Upsilon_T^p(\lambda), 
% \end{align*} 
% and 
% \begin{align*} 
% &\E \left[ d_T(\lambda)^{\sf H} \left( \frac{Y_T^{\sf H}Y_T}{N} \odot B_T \right) d_T(\lambda) \right] \\
% &= \E \left[ d_T(\lambda)^{\sf H} \left( \frac{V_T^{\sf H}V_T}{N} \odot B_T \right) d_T(\lambda) \right]\\ &+  d_T(\lambda)^{\sf H} \left( \frac{P_T^{\sf H}P_T}{N} \odot B_T \right)d_T(\lambda) \\
% &= d_T(\lambda)^{\sf H} \left( R_T \odot B_T \right) d_T(\lambda) + d_T(\lambda)^{\sf H} \left( \frac{P_T^{\sf H}P_T}{N} \odot B_T \right) d_T(\lambda)
% \end{align*} 
% where $\E \left[ d_T(\lambda)^{\sf H} \left( \frac{P_T^{\sf H}V_T}{N} \odot B_T \right)d_T(\lambda) \right] = 0$.

\end{appendix}

\bibliographystyle{IEEEtran}
\bibliography{Bibliography} 

\end{document}